\documentclass[orivec]{llncs}

\usepackage{cmap}
\usepackage{hyperref}
\usepackage{amsmath}
\usepackage{listings}
\usepackage{amsfonts}
\usepackage{graphicx}
\usepackage{courier}
\usepackage{algorithmicx}
\usepackage[table,xcdraw]{xcolor}
\usepackage{float}
\usepackage{hyperref}
\usepackage{mathtools}
\usepackage[framemethod=TikZ]{mdframed}
\usepackage[]{inputenc}
\usepackage[T1]{fontenc}
\usepackage{placeins}
\usepackage{amsmath,amssymb}
\usepackage{xspace}
\usepackage[font=footnotesize]{caption}
\usepackage[framemethod=TikZ]{mdframed}

\usepackage{fancyvrb}
\usepackage{relsize}
\graphicspath{{images/}}
\usepackage[ruled,shortend,linesnumbered,algo2e]{algorithm2e}  
\usepackage{float}
\usepackage{subfig}
\usepackage{framed}

\newcommand{\aeval}{\textsc{AE-VAL}\xspace}
\newcommand{\jkind}{\textsc{JKind}\xspace}

\newcommand{\jkindsynt}{\textsc{JSyn}\xspace} 
\newcommand{\lustrev}{\textsc{LustreV6}\xspace}

\newcommand{\isUnSat}{\textsc{isUnsat}\xspace}

\newcounter{template}
\newenvironment{template}[1][htb]
  {
   \begin{algorithm2e}[#1]%
   \SetAlgorithmName{Template}
  }{\end{algorithm2e}}

\begin{document}

\title{Synthesis from Assume-Guarantee Contracts using Skolemized Proofs of Realizability}

\author{Andreas Katis\inst{1}, Grigory Fedyukovich\inst{2},
  Andrew Gacek\inst{3}, John Backes\inst{3},\\ Arie Gurfinkel\inst{4}, Michael
  W. Whalen\inst{1}}%

\institute{
Department of Computer Science and Engineering, University of Minnesota\\
\email{katis001@umn.edu, whalen@cs.umn.edu}
\and
University of Washington Paul G. Allen School of Computer Science \& Engineering\\
\email{grigory@cs.washington.edu}
\and
Rockwell Collins Advanced Technology Center\\
\email{\{andrew.gacek,john.backes\}@rockwellcollins.com}
\and
Department of Electrical and Computer Engineering, University of Waterloo\\
\email{agurfinkel@uwaterloo.ca}
}
\maketitle


\pagestyle{plain}
\begin{abstract}
The realizability problem in requirements engineering is to determine the
existence of an implementation that meets the given formal requirements.
A step forward after realizability is proven, is to construct such an
implementation automatically, and thus solve the problem of program synthesis.
In this paper, we propose a novel approach to
program synthesis guided by k-inductive proofs of
realizability of assume-guarantee contracts constructed from safety properties. 
The proof of realizability is performed over a set of $\forall\exists$-formulas, and synthesis is performed by extracting Skolem functions witnessing the existential quantification. These Skolem functions can then be combined into an implementation.
Our approach is implemented in the \jkindsynt tool which constructs Skolem functions from a contract written in a variant of the Lustre programming language and then compiles the Skolem functions into a C language implementation.
For a variety of benchmark models that already contained hand-written
implementations, we are able to identify the usability and effectiveness of the
synthesized counterparts, assuming a component-based verification framework.

\end{abstract}

\section{Introduction}

Automated synthesis research is concerned with discovering efficient algorithms to construct programs that are guaranteed to comply with predefined temporal specifications.
This problem has been well studied for propositional specifications, especially for (subsets of) LTL~\cite{gulwani2010dimensions}.
More recently, the problem of synthesizing programs for richer theories has been examined, including work in {\em template-based synthesis}~\cite{srivastava2013template}, which attempts to find programs that match a certain shape (the template), and {\em inductive synthesis}, which attempts to use counterexample-based refinement to solve synthesis problems~\cite{flener2001inductive}.  Such techniques have been widely used for stateless formulas over arithmetic domains~\cite{reynoldscounterexample}.
\textit{Functional synthesis} has also been effectively used to synthesize
subcomponents of already existing partial
implementations~\cite{kuncak2013functional}.

In this paper, we propose a new approach that can synthesize programs for
arbitrary {\em assume/guarantee contracts} that do not have to conform to specific template
shapes or temporal restrictions and supports infinite-state representations.
The contracts are described using safety properties involving real arithmetic (LRA).  Although the
technique is not guaranteed to succeed or terminate, we have used it to successfully synthesize a range of programs over non-trivial contracts.
It is more general than previous approaches for temporal synthesis involving theories, supporting both arbitrary safety properties rather than ``stateless''
properties (unlike~\cite{reynoldscounterexample,jiang2009interpolating}) and
arbitrary shapes for synthesized programs (unlike~\cite{beyene2014constraint,srivastava2013template}).
It also natively supports the use of infinite theories, in contrast to
work on reactive synthesis from GR(1) specifications
that focused on finite-domains~\cite{walker2014predicate}, and is inherently
faster in synthesizing implementations than user-guided techniques~\cite{ryzhykdeveloping}.

Our approach extends an existing algorithm that determines
the {\em realizability} of contracts involving infinite theories such as linear
integer/real arithmetic and/or uninterpreted
functions~\cite{Katis15:Realizability,katis2015machine}.  The algorithm,
explained in Section~\ref{sec:synthesis}, uses a quantified variant of
k-induction that can be checked by any SMT solver that supports quantification. 
Notionally, it checks whether a sequence of states that satisfy the contract of
depth $k$ is sufficient to guarantee the existence of a successor state that complies
with the contract, given an arbitrary input.  An outer loop of the algorithm
increases $k$ until either a solution or counterexample is found.

The step from realizability to synthesis involves moving from the existence of a witness (as can be provided by an SMT solver such as \textsc{Z3} or \textsc{CVC4}) to the witness itself.  For this, the most important obstacle is the (in)ability of the SMT solver to handle higher-order quantification. Fortunately, interesting directions to solving this problem have already surfaced, either by extending an SMT solver with native synthesis capabilities~\cite{reynoldscounterexample}, or by providing external algorithms that reduce the problem by efficient quantifier elimination methods~\cite{fedyukovich2015automated}.

Our synthesis relies on the implementation for realizability checking
in the \jkind model checker~\cite{jkind} and the skolemization procedure
implemented in the \aeval solver~\cite{fedyukovich2015automated}.

We combined the above ideas to create a sequential synthesis
approach, which we call \jkindsynt.  It applies the realizability checker
from~\cite{Katis15:Realizability} and then extracts a Skolem witness formula
from \aeval that can immediately be turned into a C program.  We first
proposed this idea in a workshop paper~\cite{katis2016towards}; however the
paper does not contain a construction algorithm (only a notional idea of what the result would look like), proof of correctness, implementation, or experimental results.  In fact, the idea as proposed is unimplementable because it uses a Skolem {\em relation} as a witness which does not yield a functional C program.  To support synthesis, we developed a post-quantifier-elimination algorithm on top of the existing \aeval solution.
It allows producing Skolem {\em functions} rather than {\em relations}, and it is always guaranteed to terminate.

The main contributions of our work are, therefore:
\begin{itemize}
	\item To the best of our knowledge, the first template-free approach that
	effectively uses a k-inductive proof of realizability to synthesize implementations
	from safety specifications.
	\item A framework for extracting fine-grained Skolem functions for
	$\forall\exists$-formulas in linear arithmetic.
	\item A prototype tool implementing the synthesis algorithm.
	\item An experiment demonstrating the application of the tool on various
	benchmark contracts with proofs of length $k=0$ and $k=1$.
\end{itemize}

Sections~\ref{sec:synthesis} and~\ref{sec:aeval} provide the necessary formal
background used in the synthesis algorithm and \aeval, as well as adjustments to
the latter to better support the needs of this work.
Section~\ref{sec:impl} provides a detailed description regarding the
algorithm implementation. Section~\ref{sec:experiment} presents the results from
using the algorithm to synthesize leaf-level component implementations.
In Section~\ref{sec:related}, we give a brief historical background on the
related research work on synthesis. Finally, we discuss potential future work
and conclude in Section~\ref{sec:conclusion}.

\newcommand{\realizable}{\textsc{realizable}}
\newcommand{\unrealizable}{\textsc{unrealizable}}
\newcommand{\skolems}{\textit{Skolems}}
\newcommand{\init}{\textit{Init}}
\newcommand{\isValid}{\textsc{isValid}\xspace}
\newcommand{\isInvalid}{\textsc{isInvalid}\xspace}
\newcommand{\isUnsat}{\textsc{isUnsat}\xspace}

\vspace{-.5em}
\section{Synthesis from Assume-Guarantee Contracts}
\label{sec:synthesis}
\vspace{-.5em}

%
In this section we define  
Assume-Guarantee contracts (Sect.~\ref{sec:pre}),
 sketch an existing algorithm for determining contract realizability
(Sect.~\ref{sec:old}),
present our new algorithm that bridges the gap between the realizability checking and synthesis (Sect.~\ref{sec:realizability-synthesis}),
and finally illustrate how the algorithm works on an example
(Sect.~\ref{sec:example}).

\vspace{-.5em}
\subsection{Assume-Guarantee Contracts}
\label{sec:pre}
\vspace{-.5em}

One popular way to describe software requirements is through Assume-Guarantee
contracts, where requirements are expressed using safety properties that are
split into two categories, \emph{assumptions} and \emph{guarantees}.
Contract \emph{assumptions} are properties that restrict the set of valid inputs
a system can process, while \emph{guarantees} dictate system behavior by constraining system outputs.

For example, consider the contract with the assumption $A = \{x\neq
y\}$ and the guarantee $G = \{x \leq y \Longrightarrow z =
\textit{true}, x \geq y \Longrightarrow z = \textit{false}\}$, for a component with two inputs $x$ and $y$ and one output $z$.  By assumption, $x \neq y$, so the implemented system could set $z$ to true if $x < y$ and false otherwise.  
Of course, multiple implementations may exist for the same contract. An
alternative approach, for example, could set $z$ to false if $x > y$, and true
otherwise. Determining whether an implementation can be constructed to satisfy
the contract for all possible input sequences is the \emph{realizability} problem, while automatically constructing a witness of the proof of realizability of the contract is the \emph{program synthesis} problem.  The contract $(A,G)$ above is obviously \emph{realizable}, and therefore an implementation can be constructed.
However, if the assumption is omitted then the contract is \emph{unrealizable}, since there is no correct value for $z$ when $x=y$.


We describe a system using the disjoint sets $state$ and $inputs$.
Formally, an \emph{implementation} is a \emph{transition system}
described by an initial state predicate $I(s)$ of type $state \to
bool$ and by a transition relation $T(s,i,s')$ of type $state \to
inputs \to state \to bool$.

An Assume-Guarantee (AG) contract can be formally defined by a set of
\emph{assumptions} and a set of \emph{guarantees}. The
\emph{assumptions}, $A: state \rightarrow inputs \rightarrow bool$,
impose constraints over the inputs which may be modal in terms of the
previous state. The \emph{guarantees} $G$ consist of two separate
subsets $G_I: state \rightarrow bool$ and $G_T: state \rightarrow
inputs \rightarrow state \rightarrow bool$, where $G_I$ defines the
set of valid initial states, and $G_T$ specifies the properties that
need to be met during each new transition between two states. Note
that we do not necessarily expect that a contract would be defined
over all variables in the transition system, but we do not make any
distinction between internal state variables and outputs in the
formalism. This way, we can use state variables to (in some cases)
simplify the specification of guarantees.

\vspace{-.5em}
\subsection{Realizability of Contracts}
\label{sec:old}
\vspace{-.5em}

The synthesis algorithm proposed in this paper is built on top of our previous
work on a realizability checking algorithm~\cite{Katis15:Realizability}.
Using the formal foundations described in Sect.~\ref{sec:pre}, the problem of realizability is expressed using the notion of a state being \emph{extendable}:

\begin{definition}[One-step extension]
\label{def:extend}
A state $s$ is extendable after $n$ steps, denoted $\mathit{Extend}_{n}(s)$, if
any valid path of length $n-1$ starting from $s$ can be extended in response to
any valid input.%
\begin{multline*}%
\mathit{Extend}_{n}(s) \triangleq \forall i_1, s_1, \ldots, i_n, s_n.\\ A(s, i_1) \land G_T(s, i_1, s_1)
\land \cdots \land
A(s_{n-1}, i_n) \land G_T(s_{n-1}, i_n, s_n)
\implies \\
\forall i.~ A(s_n, i) \implies \exists s'.~ G_T(s_n, i, s')
\end{multline*}
\end{definition}

The algorithm for realizability uses Def.~\ref{def:extend} in two
separate checks that correspond to the two traditional cases exercised
in k-induction. Initially, it is proved that the set of initial states is
not empty, by checking for the existence of at least one
state that satisfies $G_I$. For the $\mathit{BaseCheck}$, we ensure
that all initial states are extendable for any path of length $k < n$,
while the inductive step of $\mathit{ExtendCheck}$ tries to prove that
all valid states are extendable for any path of length $n$. Therefore,
we attempt to find the smallest $n$, for which the two following
$\forall\exists$-formulas are valid:%
\begin{equation}
\label{eq:sbcheck}
\mathit{BaseCheck}(n) \triangleq \forall k < n. (\forall s. G_I(s)
	  	\implies \mathit{Extend}_k(s))
\end{equation}%
\begin{equation}
\label{eq:echeck}
\mathit{ExtendCheck}(n) \triangleq \forall s. \mathit{Extend}_n(s)
\end{equation}

The realizability checking algorithm has been used to effectively find cases
where the traditional consistency check (i.e. the existence of an assignment
to the input variables for which the output variables satisfy the contract)
failed to detect conflicts between stated requirements in case studies of
different complexity and importance. It has also been formally verified using the Coq proof assistant in terms of its
soundness, for the cases where it reports that a contract is
realizable~\cite{katis2015machine}.

\subsection{Program Synthesis from Proofs of Realizability}
\label{sec:realizability-synthesis}

The main contribution of the paper is an algorithm for deriving implementations from the proof of a contract's realizability.
Indeed, the algorithm sketched in Sect.~\ref{sec:old}
can be further used for solving the more complex problem of
\emph{program synthesis}. 
Consider checks~\eqref{eq:sbcheck}
and~\eqref{eq:echeck} that are used in the realizability checking
algorithm. Both checks require that the reachable states explored are
extendable using Def.~\ref{def:extend}. The key insights are then 1)
we can start with a arbitrary state in $G_I$ since it is non-empty, 2)
we can use witnesses from the proofs of $\mathit{Extend}_k(s)$ in
$\mathit{BaseCheck}$ to create a valid path of length $n-1$, and 3) we
can extend that path to arbitrary length by repeatedly using the
witness of the proof of $\mathit{Extend}_n(s)$ in
$\mathit{ExtendCheck}(n)$.

In first order logic, witnesses for valid $\forall\exists$-formulas
are represented by Skolem functions. Intuitively, a Skolem function
expresses a connection between all universally quantified variables in
the left-hand side of the $\forall\exists$-formulas~\eqref{eq:sbcheck}
and~\eqref{eq:echeck} and the existentially quantified variable $s'$
within $\mathit{Extend}$ on the right-hand side. Our algorithm uses
the \aeval tool, detailed in Sect.~\ref{sec:aeval}, to generate such
Skolem functions from the validity of~\eqref{eq:sbcheck}
and~\eqref{eq:echeck}.



\begin{figure}[t!]
\vspace{-5em}
\begin{minipage}{0.65\textwidth}
\scalebox{0.8}{
\begin{algorithm2e}[H]
\SetAlgoSkip{}
\SetKwFor{For}{for}{do}{}
\KwOut{$Result: \{\realizable, \unrealizable\}$, 
\skolems: Skolem list
}
\BlankLine
$Skolems \gets \langle \rangle$; \\
$InitResult \gets $\sc{Sat?}$(G_I)$; \\
\uIf(\label{alg:initState}){$(\isUnsat(InitResult))$}
	{%
		\Return
		\unrealizable, $\emptyset$, $\langle \rangle$;%
	}
\For{$(i \gets 0; \mathbf{true}; i \gets i + 1)$}{	
$ExtendResult \gets \textsc{\aeval}(ExtendCheck(i))$;
\\
\uIf(\label{alg:returnSat}){$(\isValid(ExtendResult))$}
{
	$Skolems.Add(ExtendResult.Skolem)$;\\
	\Return \realizable, 
	 \skolems;
}
$BaseResult \gets \textsc{\aeval}(BaseCheck_{k}(i))$\;
\uIf(\label{alg:returnUnsat}){$(\isInvalid(BaseResult))$}
	{%
		\Return
		\unrealizable, $\emptyset$, $\langle \rangle$;%
	}
$Skolems.Add(BaseResult.Skolem)$;\\

}
\caption{Synthesis from AG-Contracts}
\label{alg:synthesis}
\end{algorithm2e}}
\end{minipage}
\hspace{-0.8cm}
\begin{minipage}[t]{0.38\textwidth}
\scalebox{.8}{
\begin{template}[H]
\SetAlgoSkip{}
\SetKwFor{While}{forever}{do}{}
\BlankLine
  \textsc{assign\_Init()};

\BlankLine
  \textsc{read\_inputs()}\; 		
  \textsc{Skolems}[0]()\;
  $\ldots$\\
  \textsc{read\_inputs()}\;
  \textsc{Skolems}[k-1]()\;

\BlankLine

\While{}{
 \textsc{read\_inputs()}\;
 \textsc{Skolems}[k]()\;
 \textsc{update\_history()};
}
\caption{Structure of an implementation}
\label{alg:synt}
\end{template}}
\end{minipage}
\vspace{-2.5em}
\end{figure}

Alg.~\ref{alg:synthesis} provides a summary of the synthesis
procedure. The algorithm first determines whether the set of initial states $G_I$ is non-empty.
Second, it attempts to construct an inductive proof of the system's
realizability, using \aeval\ to find Skolem witnesses.  The algorithm
iteratively proves $\mathit{BaseCheck_k(i)} \triangleq \forall s. G_I(s) \implies \mathit{Extend}_i(s)$ and
accumulates the resulting Skolem functions. If
$\mathit{BaseCheck_k(i)}$ ever fails, we know $\mathit{BaseCheck(i)}$
would also fail and so the system is unrealizable. At the same time,
the algorithm tries to prove $\mathit{ExtendCheck(i)}$. As soon as the
inductive step of $\mathit{ExtendCheck(i)}$ passes, we have a complete
k-inductive proof stating that the contract is realizable. We then
complete our synthesis procedure by generating a Skolem function that
corresponds to the inductive step, and return the list of the Skolem
functions.  Note that in Alg.~\ref{alg:synthesis} for a particular depth $k$,
we perform the extends check prior to the base check.
The intuition is that $\mathit{BaseCheck(i)}$ checks
$\forall k < i$; thus, it is one step ``smaller'' than the extends
check and this avoids a special case at $k=0$.

Given a list of Skolem functions, it remains to plug them into
an implementation skeleton as shown in Template~\ref{alg:synt}.
The combination of Lustre models and k-inductive proofs
allow the properties in the model to manipulate the
 values of variables up to $k-1$ steps in the past. Thus,
the first step of an implementation  (method \textsc{assign\_Init()})
 creates an array for each state variable of size $k+1$, where
$k$ is the depth of the solution to Alg.~\ref{alg:synthesis}.
This array represents the depth of history necessary to compute
the recurrent Skolem function produced by the $\mathit{ExtendCheck}$ process.
The $\mathit{BaseCheck}$ Skolem functions initialize this history.

In each array, the $i$-th element, with $0\leq i \leq k-1$,
corresponds to the value assigned to the variable after the call to
$i$-th Skolem function. As such, the first $k-1$ elements of each array
correspond to the $k-1$ Skolem functions produced by the
$\mathit{BaseCheck}$ process, while the last element is used by the
Skolem function generated from the formula corresponding to the
$\mathit{ExtendCheck}$ process.

The template uses the Skolem functions generated by \aeval for each
of the $\mathit{BaseCheck}$ instances to describe the initial behavior of
the implementation prior to depth $k$.  
There are two ``helper'' operations:
\textsc{update\_history()} shifts each element in the arrays one position
forward (the 0-th value is simply forgotten), and \textsc{read\_inputs()} reads the
current values of inputs into the $i$-th element of the variable arrays,
where $i$ represents the $i$-th step of the process.
Once the history is entirely initialized using the $\mathit{BaseCheck}$ Skolem functions,
we add the Skolem function for the $\mathit{ExtendCheck}$ instance to describe the
recurrent behavior of the implementation (i.e., the next value of outputs in
each iteration in the infinite loop).

To establish the correctness of the algorithm,
we have constructed machine-checked proofs as to the validity of $\mathit{BaseCheck(n)}$ and
$\mathit{ExtendCheck(n)}$ using the Skolem functions.
The entirety of the models explored in
this paper only involved proofs of realizability of length $k$ equal to 0 or
1.
As such, we limited our proofs of soundness to these two specific cases. We will
extend the proofs to capture any arbitrary $k$ as part of our future work.
The theorems were written and proved using the Coq proof
assistant.~\footnote{The proofs can be found
at~\href{https://github.com/andrewkatis/Coq}{https://github.com/andrewkatis/Coq}}.

\begin{theorem}[Bounded Soundness of BaseCheck and ExtendCheck using Skolem
Functions] Let $\mathit{BaseCheck}_{S(s_n,i,s')}(n)$ and
$\mathit{ExtendCheck}_{S(s_n,i,s')}(n)$, $n \in {0,1}$, be the valid
variations of the corresponding formulas $\mathit{BaseCheck(n)}$ and
$\mathit{ExtendCheck(n)}$, where the existentially quantified part $\exists s'.~
G_T(s_n, i, s')$ has been substituted with a witnessing Skolem function
$S(s_n,i,s')$.  We have that:
\begin{itemize}
\item $\forall (A,G_{I},G_{T}). \mathit{BaseCheck}(n) \Rightarrow \mathit{BaseCheck}_{S(s_n,i,s')}(n)$
\item $\forall (A,G_{I},G_{T}). \mathit{ExtendCheck}(n) \Rightarrow
ExtendCheck_{S(s_n,i,s')}(n)$
\end{itemize}
\end{theorem}
\begin{proof}
The proof uses the definition $\mathit{Extend}_n(s)$ of an extendable state,
after replacing the next-step states with corresponding Skolem functions. From there,
the proof of the two implications is straightforward.
\qed
\end{proof}

\subsection{An Illustrative Example}
\label{sec:example}

\begin{figure}[t!]
\centering
\begin{minipage}[c]{0.6\textwidth}
 \begin{Verbatim}[fontsize=\scriptsize]
node top(x : int; state: int) returns (  );
var
  bias : int;
  guarantee1, guarantee2, guarantee3,
  guarantee4, guarantee5, guarantee_all : bool;
  bias_max : bool;
let
  bias = 0 -> (if x = 1 then 1 else -1) + pre(bias);
  bias_max = false ->
	(bias >= 2 or bias <= -2) or pre(bias_max);
  assert (x = 0) or (x = 1);
  guarantee1 = (state = 0 => (bias = 0));
  guarantee2 = true ->
  	(pre(state = 0) and x = 1) => state = 2;
  guarantee3 = true ->
  	(pre(state = 0) and x = 0) => state = 1;
  guarantee4 = bias_max => state = 3;
  guarantee5 = state = 0 or state = 1
                    or state = 2 or state = 3;
  guarantee_all = guarantee1 and guarantee2 and
          guarantee3 and guarantee4 and guarantee5;
  --%PROPERTY guarantee_all;
  --%REALIZABLE x;
tel;
 \end{Verbatim}
\end{minipage}
\scalebox{.7}{
\begin{minipage}[H]{0.5\textwidth}
\centering
\includegraphics[width=\textwidth,height=0.8\textheight,keepaspectratio]{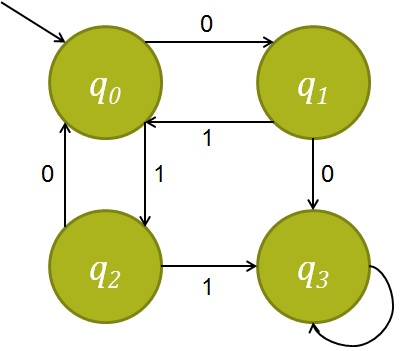}
\end{minipage}}
\vspace{-0.5em}
\caption{Requirements and possible implementation for example}
\vspace{-1.5em}
\label{fg:example}
\end{figure}

The left side of Fig.~\ref{fg:example} shows a (somewhat contrived) contract for a system that detects whether a string of two zeros or two ones ever occurs in a stream of inputs written in a dialect of the Lustre language~\cite{lustrev6}.  The right side shows a possible implementation of that contract, visualized as a state machine.
Rather than use this implementation, we would like to synthesize a new one directly from the contract. There are two unassigned variables in the contract,
\texttt{x} and \texttt{state}.
The \texttt{{-}{-}\%REALIZABLE} statement specifies that \texttt{x} is a system
input, and by its absence, that \texttt{state} is a system output. The contract's 
assumption is specified by the \texttt{assert} statement and restricts the allowable input values of \texttt{x} to either 0 or 1. We also have five guarantees:
\texttt{guarantee2} and \texttt{guarantee3} are used to indirectly
describe some possible transitions in the automaton;\footnote{In Lustre, the
arrow (\texttt{->}) and \texttt{pre} operators are used to provide an initial value and access the previous value of a stream, respectively.} \texttt{guarantee5} specifies the range of
values of variable \texttt{state};
\texttt{guarantee1} and \texttt{guarantee4} are the requirements with respect to
two local variables, \texttt{bias} and \texttt{bias\char`_max}, where
 \texttt{bias} calculates the number of successive ones or
zeros read by the automaton and \texttt{bias\char`_max} indicates that at least two zeros or two ones have been read in a row.

Note that while Lustre is a compilable language, using standard compilation tools the ``program'' in Fig.~\ref{fg:example} would not compile into a meaningful implementation: it has no outputs!  Instead, it defines the guarantees we wish to enforce within the controller, and our synthesis tool will construct a program which meets the guarantees.

The realizability check on this example succeeds with a k-inductive
proof of length $k = 1$. The two corresponding
$\forall\exists$-formulas ($k=0$ for the base check and $k=1$ for the
inductive check) are valid, and thus \aeval extracts two witnessing
Skolem functions that effectively describe assignments to the local
variables of the specification, as well as to \texttt{state} (see
Appendix~\ref{app:ex} for the particular formulas).

The Skolem functions are used to construct the final implementation
following the outline provided in Template~\ref{alg:synt}.
The main idea is to redefine each variable in the model
as an array of size equal to $k$ and
to use the $k$-th element of each array as the corresponding output of the call
to $k$-th Skolem function. After this initialization process, we use an infinite
loop to assign new values to the element corresponding to the last Skolem
function, to cover the inductive step of the original proof. The final code, a
snippet of which is presented below, is 144 lines long.
Since each Skolem is represented by an $\mathit{ite}$-statement (to be explained
in Sect.~\ref{sec:aeval}), each branch is further encoded into a C-code, as
shown in Figure~\ref{fg:snippet}.

\begin{figure}[t!]
\vspace{-2em}
\begin{minipage}{2.0\textwidth}
\begin{lstlisting}[basicstyle=\scriptsize,language=C]{Name=test2}
if (((x[1] == 1 && (-1 == bias[0])) || (x[1] == 0 && (1 == bias[0])))
     && !bias_max[0] && (state[0] != 0 || x[1] == 0)
     && (!state[0] != 0 || x[1] == 1)) {
  bias_max[1] = 0;
  bias[1] = 0;
  state[1] = 0;
}
\end{lstlisting}%
\end{minipage}
\vspace{-1em}
\caption{A code snippet of the synthesized implementation for the contract from Fig.~\ref{fg:example}.}
\vspace{-.5em}
\label{fg:snippet}
\end{figure}%

Notice how each variable is represented by an array in the snippet above.
We chose to use this easy to understand representation in order to effectively
store all the past $k-1$ values of each variable, that may be needed during the
construction of the k-inductive proof.

Recall that the user-defined model explicitly specifies only two transitions
(via \texttt{guarantee2} and \texttt{guarantee3}), while the set of implicitly defined transitions (via \texttt{guarantee1} and \texttt{guarantee4}) is incomplete.
Interestingly, our synthesized implementation turns all implicit transitions
into explicit ones which makes them executable and, furthermore, adds the
missing ones (e.g., as in the aforementioned snippet, from \texttt{state = 1} to \texttt{state = 0}).


\newcommand{\such}{\,.\,}
\newcommand{\vx}{\vec{x}}
\newcommand{\vy}{\vec{y}}
\newcommand{\Land}{\bigwedge}
\newcommand{\Lor}{\bigvee}
\newcommand{\mbp}{\mathit{MBP}}
\newcommand{\unsat}{\textsc{unsat}}
\newcommand{\sat}{\textsc{sat}}
\newcommand{\valid}{\textsc{valid}\xspace}
\newcommand{\invalid}{\textsc{invalid}\xspace}
\newcommand*\rfrac[2]{{}^{#1}\!/_{#2}}
\newcommand{\tuple}[1]{\langle #1 \rangle}       

\newcommand{\aevalalgorithm}{%
\vspace{-5em}
\begin{algorithm2e}[H]
\SetAlgoSkip{}
\SetKwFor{For}{for}{do}{}
\SetKw{KwContinue}{continue}
\KwIn{$S(\vx), \exists \vy \such T(\vx,\vy)$.}
\KwOut{Return value $\in \{\valid, \invalid\}$ of ${S(\vx)\!\! \implies\!\! \exists \vy \such T(\vx,\vy)}$, Skolem.}
\KwData{ models $\{m_i\}$, MBPs $\{T_{i}(\vx)\}$, local Skolems $\{\phi_i({\vx,\vy})\}$.}
\BlankLine
$\textsc{SmtAdd}(S(\vx))$\; 
\For{$(i \gets 0; \mathbf{true}; i \gets i + 1)$}{	
\lIf(\label{alg:returnUnsat}){$(\isUnSat(\textsc{SmtSolve}()))$}{\Return \valid, $\mathit{Sk}_{\vy} (\vx, \vy)$ from~\eqref{case:skolem}}
$\textsc{SmtPush}()$; 
$\textsc{SmtAdd}(T(\vx,\vy))$; \\
\lIf(\label{alg:returnSat}){$(\isUnSat(\textsc{SmtSolve}()))$}{\Return \invalid, $\varnothing$}
$m_i \gets \textsc{SmtGetModel}()$\label{alg:model};\\ 
$(T_{i},\phi_i({\vx,\vy}))\! \gets\! \textsc{GetMBP}(\vy, m_i, T(\vx,\vy)))$\label{alg:proj};\\
$\textsc{SmtPop}()$;
$\textsc{SmtAdd}(\neg {T_{i}})$\label{alg:block}; \\
}
\caption{\aeval \Big($S(\vx), \exists \vy \such T(\vx,\vy)$\Big), cf.~\cite{fedyukovich2015automated} }
\label{alg:ae_val}
\end{algorithm2e}
}

\newcommand{\localfactoralg}{%
\begin{algorithm2e}[H]
\SetAlgoSkip{}
\SetInd{0.4em}{0.4em}
\SetKwFor{ForAll}{forall}{do}{}
\SetKwFor{For}{for}{do}{}
\KwIn{$y_j \in \vy$, local Skolem relation $\phi(\vx,\vy) = \Land_{y_j \in \vy}(\psi_j(\vx,y_{j},\ldots, y_{n}))$, Skolem functions $y_{j+1} = f_{j+1}(\vx),\ldots, y_{n} = f_{n}(\vx)$.}
\KwData{Factored formula $\pi_j(\vx,y_{j}) = L_{\pi_j} \land U_{\pi_j} \land E_{\pi_j} \land N_{\pi_j}$. 
}
\KwOut{Local Skolem function $y_j = f_j(\vx)$.}
\BlankLine
\For{$(i \gets n; i > j; i \gets i - 1)$}{	
  $\psi_j(\vx,y_{j},\ldots,y_{n}) \gets \textsc{Substitute}(\psi_j(\vx,y_{j},\ldots,y_{n}) , y_i, f_i(\vx))$\label{alg:loc_subst};\\
}
\BlankLine

$\pi_j(\vx,y_{j}) \gets \psi_j(\vx,y_{j},\ldots,y_{n})$\label{alg:elim_compl};\\
\BlankLine

\lIf(\label{alg:trivcase}){$(|E_{\pi_j}| \neq 0)$}{\Return $E_{\pi_j}$}

$\pi_j(\vx,y_{j}) \gets \textsc{Merge}(L_{\pi_j}, \mathit{MAX}, \pi_j(\vx,y_{j}))$\label{alg:merge1};\\
$\pi_j(\vx,y_{j}) \gets \textsc{Merge}(U_{\pi_j} , \mathit{MIN}, \pi_j(\vx,y_{j}))$\label{alg:merge2};\\
\BlankLine

\uIf(){$(|N_{\pi_j}| = 0)$}{
\lIf(){$(|L_{\pi_j}| \neq 0 \land |U_{\pi_j}| \neq 0)$}{\Return $\textsc{Rewrite}(L_{\pi_j}  \land U_{\pi_j} , \mathit{MID}, \pi_j(\vx,y_{j}))$\label{alg:rewrite_mid}}
\lIf(){$(|L_{\pi_j}| = 0)$}{\Return $\textsc{Rewrite}(U_{\pi_j} , \mathit{LT}, \pi_j(\vx,y_{j}))$\label{alg:rewrite_lt}}
\lIf(){$(|U_{\pi_j}| = 0)$}{\Return $\textsc{Rewrite}(L_{\pi_j} , \mathit{GT}, \pi_j(\vx,y_{j}))$\label{alg:rewrite_gt}}
}
\BlankLine
\lElse{\Return $\textsc{Rewrite}(L_{\pi_j} \land U_{\pi_j} \land N_{\pi_j}, \mathit{FMID}, \pi_j(\vx,y_{j}))$\label{alg:rewrite_fmid}}
\caption{$\textsc{ExtractSkolemFunction}(y_j, \phi(\vx,\vy)$)}
\label{alg:loc}
\end{algorithm2e}
}

\newcommand{\skolemcases}{%
\begin{equation}
\label{case:skolem}
\mathit{Sk}_{\vy} (\vx, \vy) \equiv
\begin{cases}
  \phi_{1} (\vx, \vy)  & \text{if } T_1 (\vx) \\
    \phi_{2} (\vx, \vy)  & \text{else if } T_2 (\vx)\\
  \cdots &\text{\qquad else }\cdots \\
  \phi_{n} (\vx, \vy) & \text{\qquad\qquad else } T_n (\vx) \\
\end{cases}
\end{equation}
}

\vspace{-1em}
\section{Witnessing existential quantifiers with \aeval}
\label{sec:aeval}


An important factor for the purposes of this work, is the use of the \aeval
Skolemizer.
We start this section with a brief background on \aeval (in Sect.~\ref{sim:check}) and continue (in Sect.~\ref{sec:new}) by presenting the key contributions on delivering Skolem functions appropriate for the program synthesis from proofs of realizability.

\subsection{Validity and Skolem extraction}
\label{sim:check}

Skolemization  is a well-known
technique for removing existential quantifiers in first order formulas.
%
Given a formula ${ \exists y \such \psi(\vx, y)}$,
a~\emph{Skolem function} for $y$, $\mathit{sk}_{y}(\vx)$ is a function such that
$\exists y \such \psi(\vx,y)\!\iff\!\psi (\vx, \mathit{sk}_{y} (\vx))$.
We generalize the definition of a Skolem function for the case of a
vector of existentially quantified variables $\vy$, by relaxing the
relationships between elements of $\vx$ and $\vy$.
Given a formula ${\exists \vy \such \Psi(\vx, \vy)}$, a~\emph{Skolem relation} for $\vy$ is a relation ${\mathit{Sk}_{\vy} (\vx, \vy)}$ such that 1) $\mathit{Sk}_{\vy} (\vx, \vy) \implies \Psi (\vx, \vy)$ and 2) $\exists \vy \such \Psi(\vx, \vy)\!\iff\!\mathit{Sk}_{\vy} (\vx, \vy)$.

The pseudocode of the \aeval algorithm that  decides validity
and  extracts Skolem relation is shown in Alg.~\ref{alg:ae_val} (we refer the reader to~\cite{fedyukovich2015automated} for
more detail).
It assumes that the formula $\Psi$ can be transformed
into the form ${\exists \vy \such \Psi(\vx, \vy)} \equiv {S(\vx) \!\implies\! \exists \vy \such T(\vx,\vy)}$, where $S(\vx)$ has only existential quantifiers, and $T(\vx, \vy)$ is quantifier-free.
To decide the validity, \aeval partitions
the $\forall\exists$-formula and searches for a witnessing \emph{local} Skolem relation of
each partition.  \aeval iteratively constructs (line~\ref{alg:proj}) a set of Model-Based
Projections (MBPs): $T_i(\vx)$, such that (a) for each $i$,
$T_i(\vx) \!\implies\! \exists \vy \such T (\vx, \vy)$, and (b)
$S(\vx) \!\implies\! \Lor_i T_i(\vx)$.  Each MBP
$T_i(\vx)$ is connected with the local Skolem $\phi_i(\vx,\vy)$, such that
$\phi_i(\vx,\vy) \!\implies\! (T_{\vy_i}(\vx) \!\implies\!
  T(\vx,\vy))$.  \aeval relies on an external
procedure~\cite{komuravelli2014smt} to obtain MBPs for theories of Linear Real Arithmetic and Linear Integer Arithmetic.

Intuitively, each $\phi_i$ maps models of $S \land T_{i}$ to models of $T$.
Thus, a \emph{global} Skolem relation ${\mathit{Sk}_{\vy} (\vx, \vy)}$ is defined through a matching of each $\phi_i$ against the corresponding $T_{i}$:%
\skolemcases
     
\begin{figure}[t!]
\begin{center}
\resizebox{0.75\textwidth}{!}{%
  \begin{minipage}{.95\linewidth}
     \aevalalgorithm  
   \end{minipage}
}
\end{center}
\vspace{-3em}
\end{figure}

\vspace{-2.5em}
\subsection{Refining Skolem relations into Skolem functions}
\label{sec:new}

The output of the original \aeval algorithm does not fulfil our program synthesis needs due to two reasons: (1) interdependencies between $\vec{y}$-variables and (2) inequalities and disequalities in the terms of the Skolem relation. 
Indeed, in the lower-level \aeval constructs each MBP iteratively for each variable $y_j \in \vy$.
Thus, $y_j$ may depend on the variables of $y_{j+1},\ldots, y_{n}$ that are still not eliminated in the current iteration $j$.

Inequalities and disequalities in a Skolem relation are not desirable because the final implementation should contain assignments to each existentially quantified variable.
To specify the exact assignment value, the Skolem relation provided by \aeval should be post-processed to contain only equalities.

We formalize this procedure as finding a Skolem function $f_j(\vx)$ for each $y_j\in \vy$, such that $(y_j = f_j(\vx)) \!\implies\! \exists y_{j+1},\ldots,y_{n} \such \psi_j(\vx,y_{j},\ldots,y_{n})$.
An iteration of this procedure, for some $y_j$, is presented in Alg.~\ref{alg:loc}.
At the entry point, it assumes that the Skolem functions $f_{j+1}(\vx),\ldots,f_{n}(\vx)$ for variables $y_{j+1},\ldots,y_n$ are already computed.
Thus, Alg.~\ref{alg:loc} straightforwardly substitutes each appearance of variables $y_{j+1},\ldots, y_{n}$ in $\psi_j$ by $f_{j+1}(\vx),\ldots,f_{n}(\vx)$.
Once accomplished (line~\ref{alg:elim_compl}), formula
$\psi_j(\vx,y_{j},\ldots,y_{n})$ has the form $\pi_j(\vx,y_{j})$, i.e., it does
not contain variables $y_{j+1},\ldots, y_{n}$.

\begin{figure}[t!]
\begin{center}
\resizebox{0.75\textwidth}{!}{%
  \begin{minipage}{.9\linewidth}
     \localfactoralg
   \end{minipage}
}
\end{center}
\vspace{-3em}
\end{figure}

The remaining part of the algorithm aims to derive a function $f_j(\vx)$, such that $y_j = f_j(\vx)$.
In other words, it should construct a graph of a function that is embedded in a relation.
Note that \aeval constructs each local Skolem relation by conjoining the substitutions made in $T$ to produce $T_{i}$.
Each of those substitutions in linear arithmetic could be either an equality,
inequality, or disequality.
This allows us consider each $\pi_j(\vx,y_{j})$ to be of the following form:%
\begin{equation}
\pi_j(\vx,y_{j}) = L_{\pi_j} \land U_{\pi_j} \land M_{\pi_j} \land V_{\pi_j} \land E_{\pi_j} \land N_{\pi_j}
\end{equation}
where:%
\begin{align*}
L_{\pi_j} &\triangleq \Land_{l \in C({\pi_j})}(y_j > l(\vx)) &
U_{\pi_j} &\triangleq \Land_{u \in C({\pi_j})}(y_j < u(\vx)) &  
M_{\pi_j} &\triangleq \Land_{l \in C({\pi_j})}(y_j \ge l(\vx)) \\
V_{\pi_j} &\triangleq \Land_{u \in C({\pi_j})}(y_j \le u(\vx))  &
E_{\pi_j} &\triangleq \Land_{e \in C({\pi_j})}(y_j = e(\vx))  &
N_{\pi_j} &\triangleq \Land_{h \in C({\pi_j})}(y_j \neq h(\vx))  
\end{align*}

We present several primitives needed to construct $y_j = f_j(\vx)$ from $\pi_j(\vx,y_{j})$ based on the analysis of terms in $L_{\pi_j}$, $U_{\pi_j}$, $M_{\pi_j}$, $V_{\pi_j}$, $E_{\pi_j}$ and $N_{\pi_j}$.
For simplicity, we omit some details on dealing with
non-strict inequalities in $M_{\pi_j}$ and $V_{\pi_j}$ since
they are similar to strict inequalities in $L_{\pi_j}$ and $U_{\pi_j}$.
Thus, without loss of generality, we assume that $M_{\pi_j}$ and $V_{\pi_j}$ are empty.
In the description, we denote the number of conjuncts in formula $A$ as $|A|$.
Finally, due to lack of space, we focus on Linear Real Arithmetic in this section; and the corresponding routine for Linear Integer Arithmetic is worked out similarly.

The simplest case (line~\ref{alg:trivcase}) is when there is at least one conjunct $(y_j = e(\vx)) \in E_{\pi_j}$.
Then $(y_j = e(\vx))$ itself is a Skolem function.
Otherwise, the algorithm creates a Skolem function from the following primitives.

\begin{definition}
Let $l(\vx)$ and $u(\vx)$ be two terms in linear real arithmetic,
then operators $\mathit{MAX}$, $\mathit{MIN}$, $\mathit{MID}$, $\mathit{LT}$, $\mathit{GT}$ are defined as follows:
\begin{align*}
\mathit{GT} (l) (\vx) &\triangleq l (\vx) + 1 \notag &
\mathit{MAX}(l, u) (\vx) &\triangleq ite (l (\vx) < u (\vx), u(\vx), l(\vx)) \notag \\
\mathit{LT} (u) (\vx) &\triangleq u (\vx) -1 \notag &
\mathit{MIN}(l, u) (\vx) &\triangleq ite (l (\vx) < u (\vx), l(\vx), u(\vx)) \notag \\
\mathit{MID}(l, u) (\vx) &\triangleq \frac{l(\vx) + u(\vx)}{ 2} \notag 
\end{align*}
\end{definition}

\begin{lemma}
If $|L_{\pi_j}| = n$, and $n>1$, then $L_{\pi_j}$ is equivalent to $y_j > \mathit{MAX} (l_1,$ $\mathit{MAX} (l_2,\ldots \mathit{MAX} (l_{n-1}, l_n) )) (\vx)$.
If $|U_{\pi_j}| = n$, and $n>1$, then $U_{\pi_j}$ is equivalent to $y_j < \mathit{MIN} (u_1, \mathit{MIN} (u_2,\ldots \mathit{MIN} (u_{n-1}, u_n) ))(\vx)$.
\end{lemma}
This primitive is applied (lines~\ref{alg:merge1}-\ref{alg:merge2}) in order to reduce the size of $L_{\pi_j}$ and $U_{\pi_j}$.
Thus, from this point on, with out
loss of generality, we assume that each $L_{\pi_j}$ and $U_{\pi_j}$
have at most one conjunct.

\begin{lemma}
If $|L_{\pi_j}| = |U_{\pi_j}|=1$, and $|E_{\pi_j}| = |N_{\pi_j}| = 0$, then the Skolem relation can be rewritten into $y_j = \mathit{MID} (l, u)(\vx)$. 
\end{lemma}
This primitive is applied (line~\ref{alg:rewrite_mid}) in case the graph of a Skolem function can be constructed exactly in the middle of the two graphs for the lower- and the upper boundaries for the Skolem relation.
Otherwise, if some of the boundaries are missing, but $|N_{\pi_j}| = 0$ (lines~\ref{alg:rewrite_lt}-\ref{alg:rewrite_gt}), then the following primitive is applied:

\begin{lemma}
If $|L_{\pi_j}| =1$, and $|U_{\pi_j}| = |E_{\pi_j}| = |N_{\pi_j}| = 0$, then the Skolem relation can be rewritten into the form $y_j = \mathit{GT} (l)(\vx)$.
If $|U_{\pi_j}|=1$, and $|L_{\pi_j}|=|E_{\pi_j}|=|N_{\pi_j}|=0$, then the Skolem relation can be rewritten into the form $y_j = \mathit{LT} (l)(\vx)$.
\end{lemma}

Finally, the algorithm handles the cases when $|N_{\pi_j}|>0$ (line~\ref{alg:rewrite_fmid}).
We introduce operator $\mathit{FMID}$ that for given $l$, $u$, and $h$ and each $\vx$ outputs either $\mathit{MID}(l,u)$ or $\mathit{MID}(l,\mathit{MID}(l,u))$ depending on if $\mathit{MID}(l,u)$ equals to $h$ or not.

\begin{lemma}
If  $|L_{\pi_j}|=|U_{\pi_j}|=|N_{\pi_j}|=1$, and $|E_{\pi_j}|=0$, then the Skolem relation can be rewritten into the form $y_j\! =\! \mathit{FMID} (l,\! u,\! h)(\vx)$, where%
\begin{align}
\mathit{FMID}(l, u, h) (\vx) \triangleq ite \Big(&\mathit{MID}(l, u) (\vx) = h(\vx),  \notag \\
                          &\mathit{MID}\big(l, \mathit{MID}(l, u)\big) (\vx), \,\,
                          \mathit{MID}(l, u) (\vx)\Big) \notag 
\end{align}
\end{lemma}
For $|N_{\pi_j}|>1$, the Skolem gets rewritten
in a similar way recursively.

\begin{lemma}
If $|L_{\pi_j}|=|N_{\pi_j}|=1$, and $|E_{\pi_j}|=|U_{\pi_j}|=0$, then the Skolem relation can be rewritten into the form $y_j = \mathit{FMID} (l, GT(l), h)(\vx)$. 
If $|U_{\pi_j}|=|N_{\pi_j}|=1$, and $|E_{\pi_j}|=|L_{\pi_j}|=0$, then the Skolem relation can be rewritten into the form $y_j = \mathit{FMID} (LT(u), u, h)(\vx)$. 
\end{lemma}

\begin{theorem}[Soundness]
Iterative application of Alg.~\ref{alg:loc} to all variables $y_n,\ldots,y_1$ returns a local Skolem function to be used in~\eqref{case:skolem}.
\end{theorem}
\begin{proof}
Follows from the case analysis that applies the lemmas above.
\qed
\end{proof}

Recall that the presented technique is designed to effectively remove inequalities and disequalities from local Skolem relations.
The resulting local Skolem functions enjoy a more fine-grained and
easy-to-understand form.
We admit that for further simplifications (that would benefit the synthesis procedure), we can exploit techniques to rewrite $\mbp$-s into compact \emph{guards}~\cite{fedyukovich2015automated}.


\vspace{-.5em}
\section{Implementation}
\label{sec:impl}
\vspace{-.5em}

We developed \jkindsynt, our synthesis algorithm on top of \jkind~\cite{jkind},
a Java implementation of the \textsc{KIND} model
checker.%
\footnote{An unofficial release of \jkind
supporting synthesis is available to download at
https://github.com/andrewkatis/jkind-1/tree/synthesis. \aeval needs to be
installed separately from https://github.com/grigoryfedyukovich/aeval.}
Each model is described using the Lustre language, which is used as an intermediate language to formally verify contracts in the Assume-Guarantee Reasoning (\textsc{AGREE}) framework~\cite{NFM2012:CoGaMiWhLaLu}.
Internally, \jkind uses two parallel engines (for $\mathit{BaseCheck}$ and
$\mathit{ExtendCheck}$) in order to construct a $k$-inductive proof for the
property of interest.
The first order formulas that are being constructed are then fed to the \textsc{Z3} SMT
solver~\cite{DeMoura08:z3} which provides state of the art support for reasoning
over quantifiers and incremental search.
For all valid $\forall\exists$-formulas, \jkindsynt proceeds to construct a list of
Skolem functions using the \aeval Skolemizer.
\aeval supports LRA  and LIA and thus provides the Skolem relation over integers and reals.%
\footnote{For realizability checks over Linear Integer and Real
  Arithmetic (LIRA), \jkind has an option to use \textsc{Z3} directly.}

As discussed in Sect.~\ref{sec:realizability-synthesis}, we construct a Skolem
function for each base check up to depth $k$ and one for the inductive relation
at depth $k$.  What remains is to knit those functions together into an implementation in C.  The \textsc{SMTLib2C} tool performs this translation, given an input list of the original Lustre specification (to determine the I/O interface) and the Skolem functions (to define the behavior of the implementation).  The main translation task involves placing the Skolem functions into the template described in Alg.~\ref{alg:synt}.  Each Skolem function describes a bounded history of at most depth $k$ over specification variables, so each variable is represented by an array of size $k$ in the generated program.  The tool ensures that the array indices for history variables match up properly across the successive base- and inductive-case Skolem functions.  Note that during this translation process, real variables in Skolem functions are defined as floats in C, which could cause overflow and precision errors in the final implementation.  We will address this issue in future work.



\vspace{-.5em}
\section{Experimental Results}
\label{sec:experiment}
\vspace{-.5em}

For the purposes of this work, we synthesized implementations for 58
contracts written in Lustre,
including the running example from Fig.~\ref{fg:example}.~\footnote{The
benchmarks can be found at https://github.com/andrewkatis/synthesis-benchmarks/tree/master/verification.}
The original Lustre programs contained both the contract as well as an
implementation, which provided us with a complete test benchmark suite since
we were able to compare the synthesized implementations to already existing handwritten
solutions. By extracting the handwritten implementation, we synthesized an
alternative, and translated both versions to an equivalent C representation,
using \textsc{SMTLib2C} for the synthesized programs and the \lustrev
compiler~\cite{lustrev6}, including all of its optimization options, for the
original implementations.


\begin{figure}[t!]
	\centering
    \begin{tabular}[b]{c}
    \includegraphics[width=0.48\textwidth,height=\textheight,keepaspectratio]{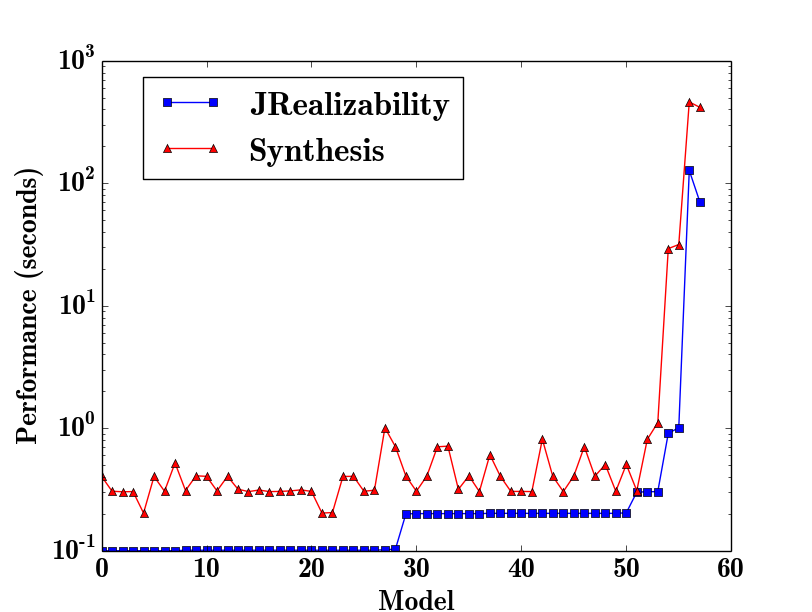}
    \\
    \small (a) \quad \quad \quad
    \label{fg:overhead} 
    \end{tabular} 
    \begin{tabular}[b]{c}
    \includegraphics[width=0.48\textwidth,height=\textheight,keepaspectratio]{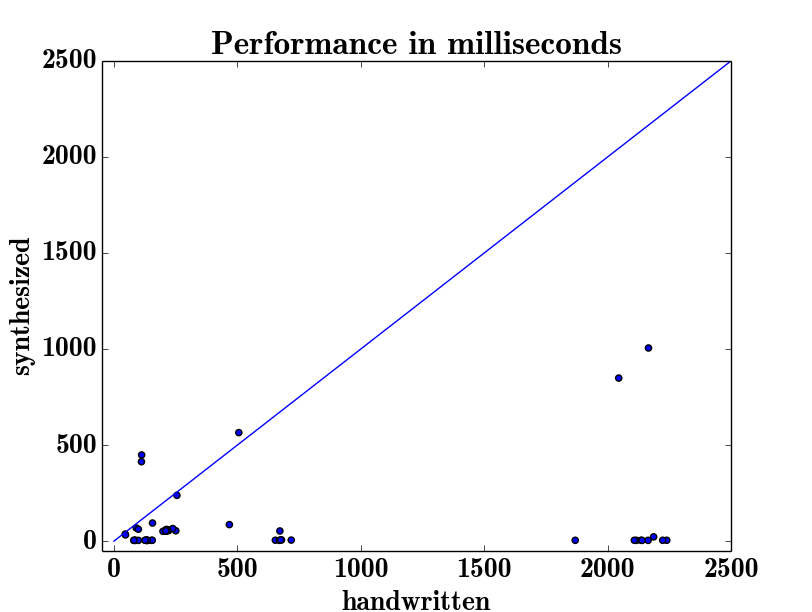}
    \\
    \small \quad \quad (b)
    \label{fg:performance} 
    \end{tabular}
    \\
\vspace{-1.5em}	
    \begin{tabular}[b]{c}
    \includegraphics[width=0.5\textwidth,height=\textheight,keepaspectratio]{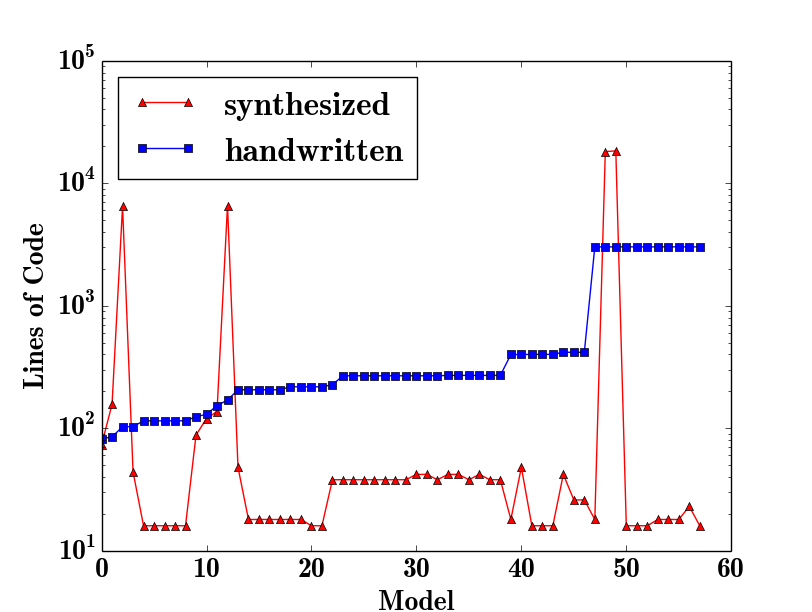}
    \\ \small (c)
    \label{fg:loc}
    \end{tabular}
	\vspace{-1em}
	\caption{Experimental results}
	\label{fig:expfigs}
	\vspace{-2em}
\end{figure}

Fig.~\ref{fig:expfigs}(a) shows the overhead of synthesis by \jkindsynt
comparing to the realizability checking by \jkind, while 
%
Fig.~\ref{fig:expfigs}(b) provides a scatter plot of the results of our
experiments in terms of the performance of the synthesized programs against the original, handwritten
implementations. Each dot in the scatter plot represents a pair of
running times (x - handwritten, y - synthesized) of the 58
programs.
For the two most complex models in the benchmark
suite, the synthesized implementations underperform the
programs generated by \lustrev. As the level of complexity decreases, we notice
that both implementations share similar performance levels, and for the most
trivial contracts in the experiment set, the synthesized programs perform better
with a noticeable gap. We attribute these results mainly to the
simplicity of the requirements expressed in the majority of the models which were proved realizable for $k=0$ by \jkind,
except for the example from Section~\ref{sec:example}  and two complex contracts for a cruise
controller, which were proved for $k=1$.
It is important at this point to recall the fact that the synthesized
implementations are not equivalent to the handwritten versions, in a similar
fashion to the example used in Section~\ref{sec:example}.

Fig.~\ref{fig:expfigs}(c) presents the size of the implementations.
Here, we can see the direct effect of the specification complexity to the size
of the Skolem functions generated by \aeval. Two out of the five synthesized
programs that are larger than their handwritten counterparts were also slower
than the handwritten implementations. Since the majority of the models contained simple requirements, the overall size of the synthesized implementation remained well below \lustrev-programs.

Handwritten implementations are still prevalent in application domains since
they provide advantages in numerous aspects, such as readability, extendability
and maintenance. Nevertheless, the results show that the synthesized
implementations can be used as efficient placeholders to reduce the
time required to verify a system under construction, without needing a final implementation for all its components.
%

\vspace{-.5em}
\section{Related Work}
\vspace{-.5em}

Pnueli and Rosner were the
first to involve the formal definition of a reactive system's realizability (or
implementability), introduce the notion of a \emph{Skolem paradigm}, as well as
describe a process to synthesize implementations for temporal
specifications~\cite{DBLP:conf/popl/PnueliR89}.	
%
%
%
Since then, a vast variety of techniques have been developed. Efficient
algorithms were proposed for subsets of propositional LTL
\cite{Klein10,ehlers2010symbolic,cheng2016structural} simple
LTL formulas \cite{Bohy12,Tini03}, as well as other temporal
logics \cite{monmege2016real,Hamza10}, such as SIS~\cite{Aziz95}.
Component-based approaches have also been explored in~\cite{Chatterjee07}.


Template-based approaches to
synthesis described in~\cite{srivastava2013template,beyene2014constraint} focus on the exploration of programs that satisfy a specification that is refined after each
iteration, following the basic principles of deductive synthesis.
In particular, the \textsc{E-HSF} engine~\cite{beyene2014constraint}, 
uses a predefined set of templates to search for potential Skolem
relations and thus to solve $\forall\exists$-formulas. 
In contrast, our synthesis algorithm is template-free.
Enumeration, is not an optimal choice, because the Skolem functions generated from
our benchmarks may contain 10-20 branches of the $ite$-statement. In these
cases, an enormous number of potential shapes may exist, thus prohibiting the use of
enumeration. 

Inductive synthesis is an active area of research where the main
goal is the generation of an inductive invariant that can be used to describe the space of programs that are guaranteed to satisfy the given specification~\cite{flener2001inductive}. This idea is mainly supported by the use of SMT solvers to guide the invariant refinement through traces that violate the requirements, known as counterexamples. Our approach differentiates from this approach by only considering the capability of constructing k-inductive proofs, with no further refinement of the problem space.

A rather important contribution in the area is the recently published work by
Ryzhyk and Walker~\cite{ryzhykdeveloping}, where they share their
experience in developing and using a reactive synthesis tool called \textsc{Termite} for
device drivers in an industrial environment. The driver synthesis uses
a predicate abstraction technique~\cite{walker2014predicate} to efficiently
cover the state space for both safety and liveness GR(1) specifications,
leveraging the theory of fixed-size vectors. The authors follow a user-guided
approach, that continuously interacts with the user in order to combat ambiguities in the specification.
In contrast, our approach supports safety specifications using infinite-state,
linear-arithmetic domains and follows a ``hands-off'', automated process.

\label{sec:related}


\vspace{-.5em}
\section{Future Work and Conclusion}
\label{sec:conclusion}
\vspace{-.5em}

In this paper, we contributed an approach to program synthesis guided by the
proofs of realizability of Assume-Guarantee contracts.
To check realizability, it performs k-induction-based reasoning to decide validity of a set of $\forall\exists$-formulas.
Whenever a contract is proven realizable, it further employs the Skolemization procedure and extracts a fine-grained witness to the realizability.
Those Skolem functions are finally encoded into a desirable implementation.
We implemented the technique in the \jkindsynt tool and evaluated it for the set
of Lustre models of different complexity.
The experimental results provided fruitful conclusions on the overall efficacy of the the approach.

To the best of our knowledge our work is the first complete attempt on providing
a synthesis algorithm based on the principle of k-induction using infinite
theories. The ability to express contracts that support ideas from many
categories of specifications, such as template-based and temporal properties,
increases the potential applicability of this work to multiple subareas on
synthesis research.


In future work, we plan to exploit the connections of our approach with
Property Directed Reachability~\cite{komuravelli2014smt} more closely. 
Another promising idea here is the use of Inductive Validity Cores (IVCs)~\cite{Ghass16}, whose main purpose is to effectively pinpoint the absolutely necessary model elements in a generated proof. We can potentially use
the information provided by IVCs as a preprocessing tool to reduce the size of
the original specification, and hopefully the complexity of the realizability
proof. Finally, various optimizations can be implemented on top of both
\aeval and \textsc{SMTLib2C} to produce smaller implementations.
\vspace{-.5em}

\bibliography{document}
\bibliographystyle{splncs03}

\newpage
\appendix
\section{Example in more detail}
\label{app:ex}

Here we consider our example from Fig.~\ref{fg:example} and demonstrate one iteration of the synthesis procedure.
%
%
%
In particular the $\forall\exists$-formula of $\mathit{ExtendCheck}$ is as follows:%
\begin{equation}
\begin{aligned}
    &\quad {}(x_0 = 0 \lor x_0 = 1) \land \\
    bias_0&=ite(init, 0, ite(x_0= 1, 1, -1)+bias_{-1}) \land \\
    bias\_max_0&=ite(init, false, ((bias_0\ge 2)\lor (bias_0\le -2))\lor bias\_max_{-1})\land \\
    guarantee_{10}&=((state_0=0)\implies(bias_0=0))\land \\
    guarantee_{20}&=ite(init, true, ((state_{-1}=0)\land x_0 = 1)\implies(state_0=2))\land \\
    guarantee_{30}&=ite(init, true, ((state_{-1}=0)\land x_0 = 0)\implies(state_0=1))\land \\
    guarantee_{40}&=(bias\_max_0\implies(state_0=3))\land \\
    guarantee_{50}&=((state_0=0) \lor (state_0=1) \lor (state_0=2) \lor (state_0=3))\land \\
    guarantee\_all_0&=(guarantee_{10}\land guarantee_{20}\land
    guarantee_{30}\land guarantee_{40} \land \\ &\qquad {} guarantee_{50}) 
    \land guarantee\_all_0 \land \\
    bias_1&=ite(false, 0, ite(x_0= 1, 1, -1)+bias_0)\land \\
    bias\_max_1&=ite(false, false, ((bias_1\ge 2)\lor (bias_1\le-2))\lor bias\_max_0)\land \\
    guarantee_{11}&=((state_1=0)\implies(bias_1=0))\land \\
    guarantee_{21}&=ite(false, true, ((state_0=0)\land x_0 = 1)\implies(state_1=2))\land \\
    guarantee_{31}&=ite(false, true, ((state_0=0)\land x_0 = 0)\implies(state_1=1))\land \\
    guarantee_{41}&=(bias\_max_1\implies(state_1=3))\land \\
    guarantee_{51}&=((state_1=0) \lor (state_1=1) \lor (state_1=2) \lor (state_1=3))\land \\
    guarantee\_all_1&=(guarantee_{11}\land guarantee_{21}\land guarantee_{31}\land guarantee_{41} \land guarantee_{51})
\implies \\ &\qquad {}
\exists bias_2, bias\_max_2, guarantee_{12}, state_2, guarantee_{22},
guarantee_{32}, \\ &\qquad {} guarantee_{42}, guarantee_{52}, guarantee\_all_2
\such  \\
    &\quad {}(x_1 = 0 \lor x_1 = 1) \land \\
  bias_2&=ite(false, 0, ite(x_1= 1, 1, -1)+bias_0) \land\\
  bias\_max_2&=ite(false, false, ((bias_2\ge 2) \lor (bias_2\le -2)) \lor bias\_max_0) \land\\
  guarantee_{12}&=((state_2=0)\implies (bias_2=0)) \land\\
  guarantee_{22}&=ite(false, true, ((state_0=0)\land x_1 = 1)\implies (state_2=2)) \land\\
  guarantee_{32}&=ite(false, true, ((state_0=0)\land x_1 = 0)\implies (state_2=1)) \land\\
  guarantee_{42}&=(bias\_max_2\implies (state_2=2)) \land\\
  guarantee_{52}&=((state_2=0) \lor (state_2=1) \lor (state_2=2) \lor (state_2=3))\land \\
  guarantee\_all_2&=(guarantee_{12}\land guarantee_{22}\land guarantee_{32}\land guarantee_{42} \land guarantee_{52})\land \\
  &\qquad {} guarantee\_all_2&\notag
\end{aligned}
\end{equation}

\aeval proceeds by constructing $\mbp$-s and creating local Skolem functions. 
In one of the iterations, it obtains the following $\mbp$:%
\begin{equation}
\begin{aligned}
  (x_1 = 1\land -1=bias_0) &\lor (x_1=0 \land 1=bias_0) \land \\
  \neg bias\_&max_0 \land \\
  (\neg (state_0=0))\lor  x_1 = 0 &\land
  (\neg (state_0=0))\lor x_1 =1  \notag
\end{aligned}
\end{equation}

and the following local Skolem function:%
\begin{align*}
  state_2&=0 \land&
  bias_2&=0 \land&
  bias\_max_2&=0
\end{align*}

In other words, the pair of the $\mbp$ and the local Skolem function is the synthesized implementation for some transitions of the automaton: the $\mbp$ specifies the source state, and the Skolem function specifies the destination state.
From this example, it is clear that the synthesized transitions are from state 1 to state 0, and from state 2 to state 0.
These $\mbp$ and the local Skolem are further encoded into the snippet of C code that after slight simplifications looks as follows:

\begin{figure}[h]
\begin{lstlisting}[language=C]{Name=test2}
...
if (((x[1] == 1 && (-1 == bias[0])) || (x[1] == 0&& (1 == bias[0])))
     && !bias_max[0] 
     && (state[0] != 0 || x[1] == 0) && (!state[0] != 0 || x[1] == 1)) {
  bias_max[1] = 0;
  bias[1] = 0;
  state[1] = 0;
}
...
\end{lstlisting}
\end{figure}
\end{document}